\documentclass[citeauthoryear]{GTM}
\usepackage{amsmath,amssymb}
\usepackage{color}
\usepackage{hyperref}

\title{Comparing the Manipulability of Approval Voting and Borda\thanks{This work is supported by the Russian Science Foundation under grant 20-71-00029.}}

\author{Daria Teplova\inst{1} \and Egor Ianovski\inst{2}}

\institute{ITMO University,\\
Kronverkskiy Prospekt 49, St Petersburg, 197101, Russia
\and
HSE University,\\
Ulitsa Soyuza Pechatnikov 16, St Petersburg, 190121, Russia}

\newcommand{\CV}{\mathcal{V}}
\newcommand{\CC}{\mathcal{C}}
\newcommand{\CL}{\mathcal{L}}
\newcommand{\ria}{\rightarrow}
\newcommand{\set}[1]{\{\,#1\,\}}
\newcommand{\ps}{\geq_{\mathcal{PS}}}
\newcommand{\pss}{>_{\mathcal{PS}}}
\newcommand{\xps}{\times_{\mathcal{PS}}}
\newcommand{\nps}{\ngeq_{\mathcal{PS}}}

\newcommand{\ceil}[1]{\lceil#1\rceil}
\newcommand{\floor}[1]{\lfloor#1\rfloor}
\newcommand{\score}{\text{score}}

\begin{document}

\maketitle

\begin{abstract}
The Gibbard-Satterthwaite theorem established that no non-trivial voting rule is strategy-proof, but that does not mean that all voting rules are equally susceptible to strategic manipulation. Over the past fifty years numerous approaches have been proposed to compare the manipulability of voting rules in terms of the probability of manipulation, the domains on which manipulation is possible, the complexity of finding such a manipulation, and others. In the closely related field of matching, \cite{Pathak2013} pioneered a notion of manipulability based on case-by-case comparison of manipulable profiles. The advantage of this approach is that it is independent of the underlying statistical culture or the computational power of the agents, and it has proven fruitful in the matching literature. In this paper, we extend the notion of Pathak and S\"onmez to voting, studying the families of $k$-approval and truncated Borda scoring rules. We find that, with one exception, the notion does not allow for a meaningful ordering of the manipulability of these rules.
\vskip3pt

\end{abstract}

\section{Introduction}

Strategy-proofness -- the idea that it is in an agent's interest to reveal their true preferences -- is a fundamental desideratum in mechanism design. All the other properties a mechanism may have become suspect if we can not assume that an agent will play according to the rules. Unfortunately, in the field of voting, this is a property we have to live without -- per the Gibbard-Satterthwaite theorem (\cite{Gibbard1973}; \cite{Satterthwaite1975}), the only strategy-proof rule with a range consisting of more than two candidates is dictatorship. This negative result, however, did not mean that scholars were willing to give up on either voting or resistance to strategy. Instead, the search was on for a workaround.

 It was already known that the impossibility does not hold on restricted domains \cite{Dumett1961}, and if the preferences of the voters are separable \cite{Barbera1991} or single-peaked \cite{Moulin1980}, then natural families of strategy-proof voting rules exist. For those committed to the universal domain, there was the statistical approach -- all rules may be manipulable, but it could be the case that some are more manipulable than others. This lead to a voluminous literature on manipulation indices, that sought to quantify how likely a voting rule is to be manipulable (\cite{Nitzan1985}; \cite{Kelly1993}; \cite{Aleskerov1999}). With the incursion of computer science into social choice, an approach based on computational complexity came into prominence -- the idea being, if a strategic vote is computationally infeasible to find, that is almost as good as there being no strategic vote in the first place (\cite{Bartholdi1989}; \cite{Conitzer2007}; \cite{Walsh2011}).
 
 None of these approaches were entirely convincing. Domain restrictions are by nature arbitrary, and there is little point in arguing as to how natural single-peaked preferences may be, if no real-world example actually is \cite{Elkind2017}. Manipulation indices are sensitive to the choice of the statistical culture, and are usually obtained by means of computer simulations for particular choices of the number of voters and candidates; so while an index could tell us which voting rule is more manipulable under, say, impartial culture with four voters and five candidates, it would be a stretch to extrapolate from this to statements about the manipulability of a voting rule in general. Computational complexity focuses on the worst case of finding a strategic vote, and a high worst-case complexity does not preclude the possibility of the problem being easy in any practical instance \cite{Faliszewski2011}.

In the closely related field of matching, \cite{Pathak2013} proposed a method to compare the manipulability of mechanisms that seemed to sidestep all these issues -- mechanism $f$ is said to be more manipulable than $g$ if the set of profiles on which $g$ is manipulable is included in the set of profiles on which $f$ is manipulable. No restrictions on domain, statistical culture, or computational ability is required. In the appendix to their paper, Pathak and S\"onmez theorised how the approach could be extended to other areas of mechanism design. This was taken up by the matching community (\cite{Decerf2021}; \cite{Bonkoungou2021}), but to our knowledge the only authors to apply this approach to voting were Arribillaga and Mass\'o (\cite{Arribillaga2016}; \cite{Arribillaga2017}). However, the notion used by Arribillaga and Mass\'o differed from that of Pathak and S\"onmez. A l\`a Pathak and S\"onmez, we would say that a voting rule $f$ is more manipulable than $g$ just if:
$$\forall P: \text{$g$ is manipulable at $P$}\Rightarrow\text{$f$ is manipulable at $P$}.$$
In other words, if a voter can manipulate $g$ in profile $P$, then a voter can also manipulate $f$ in the same profile. Arribillaga and Mass\'o's notion is a bit harder to parse:
$$\forall P_i: (\exists P_{-i}: \text{$g$ is manipulable at $(P_i,P_{-i})$})\Rightarrow (\exists P_{-i}': \text{$f$ is manipulable at $(P_i,P_{-i}')$}).$$
That is, if there exists a possible \emph{preference order} of voter $i$, $P_i$, such that there exists \emph{some} profile extending $P_i$, in which the voter can manipulate $g$, then there also exists a \emph{possibly different} profile extending $P_i$, in which the voter can manipulate $f$. To see why this could be an issue, recall that a voting rule is neutral if for any permutation of candidates $\pi$, $f(\pi P)=\pi f(P)$. A neutral voting rule is \emph{always} manipulable under the definition above -- by the Gibbard-Satterthwaite theorem, there exists some profile $P$ where a voter can manipulate. If we pick an appropriate permutation of candidate names, we will obtain a manipulable profile where voter $i$'s preference order is any order we want. Indeed, the papers of Arribillaga and Mass\'o deal with the manipulability of median voter schemes \cite{Moulin1980} and voting by committees \cite{Barbera1991}, both of which are fundamentally non-neutral procedures.

Most voting rules, however, are intended to capture democratic values, and so are neutral at their core; neutrality is typically relaxed only for the purposes of tie-breaking. The notion of Arribillaga and Mass\'o is inappropriate in this setting. The purpose of this paper is to see whether the original notion of Pathak and S\"onmez is any better.

\subsection{Our contribution}

We apply the manipulability notion of Pathak and S\"onmez to the families of $k$-approval and truncated Borda voting rules. We find that all members of the $k$-approval family are incomparable with respect to this notion, while in the truncated Borda family, in the special case of two voters, $(k+1)$-Borda is more manipulable than $k$-Borda; all other members are incomparable.

\section{Preliminaries}

Let $\CV$, $|\CV|=n$, be a set of voters, $\CC$, $|\CC|=m$, a set of candidates,  and $\CL(\CC)$ the set of linear orders over $\CC$. Every voter is associated with some $\succeq_i\in\CL(\CC)$, which denotes the voter's preferences. A profile $P\in\CL(\CC)^n$ is an $n$-tuple of preferences, $P_i$ is the $i$th component of $P$ (i.e.\ the preferences of voter $i$), and $P_{-i}$ the preferences of all the other voters.

A voting rule is a mapping:
$$f:\CL(\CC)^n\ria \CC$$

Note two consequences of the definition above. First, the number of voters and candidates is integral to the definition of a voting rule. I.e., for the purposes of this paper the Borda rule with $n=3, m=4$ is considered to be a different voting rule from the Borda rule with $n=4, m=3$. This is why our results meticulously consider every combination of $n$ and $m$ in detail.

Second, since we are requiring the voting rule to output a single candidate, we are assuming an inbuilt tie-breaking mechanism. For the purposes of this paper, all ties will be broken lexicographically. Capital letters will be used to denote candidates with respect to this tie-breaking order. That is, in a tie between $A$ and $B$ the tie is broken in favour of $A$. In the case of subscripts, ties are broken first by alphabetical priority then by subscript. That is, in the tie $\set{A_3,A_5,B_1}$, the winner is $A_3$ since $A$ has priority over $B$ and $3$ is smaller than 5.

In cases where we do not know a candidate's position in the tie-breaking order, we denote the candidate with lower case letters. Thus, if the tie is $\set{a,b,c}$, we cannot say who the winner is, and must proceed by cases.

We study two families of voting rules:

\begin{definition}
$k$-approval, denoted $\alpha_k$, is the voting rule that awards one point to a candidate each time that candidate is ranked in the top $k$ positions of a voter. The highest scoring candidates are the tied winners, ties are broken lexicographically.

$k$-Borda, denoted $\beta_k$, is the voting rule that awards $k-i+1$ points to a candidate each time that candidate is ranked $i$th, $i\leq k$. The highest scoring candidates are the tied winners, ties are broken lexicographically.

The corner case of $\alpha_1=\beta_1$ is also known as the plurality rule, while $\beta_{m-1}$ is known as the Borda rule.

Both families are special cases of \emph{scoring rules}, under which a candidate gains $s_i$ points each time it is ranked $i$th. Under $k$-approval, $s_1=\dots=s_k=1$, while under $k$-Borda $s_i=\max(0,k-i+1)$.
\end{definition}

We will be comparing the $k$-approval and $k$-Borda families of voting rules via the notion of manipulability pioneered by \cite{Pathak2013}.

\begin{definition}[Pathak and S\"onmez, 2013]
Let $f,g$ be two voting rules. We say that $f$ is manipulable at profile $P$ just if there exists a voter $i$ and a preference order $P_i'$ such that:
$$f(P_i',P_{-i})\succ_i f(P_i,P_{-i}).$$

We say that $f$ is more manipulable than $g$, denoted $f\ps g$, just if, for every profile $P$, if $g$ is manipulable at $P$ then so is $f$.

$f\pss g$ is shorthand for $f\ps g$ and $g\nps g$, and $f\xps g$ is shorthand for $f\nps g$ and $g\nps f$.
\end{definition}
\section{$k$-Approval family}

In this section we fix $i<j$. Our end result (\autoref{thm:approval}) is that for any $n,m$, $\alpha_i\xps\alpha_j$.

\begin{proposition}\label{lemLaij}
For all $n,m$: $\alpha_{i} \nps \alpha_{j}$.
\end{proposition}
\begin{proof}
Consider a profile with $i$ $B$ candidates, $j-i$ $A$ candidates, and $m-j$ $C$ candidates.

\[
\begin{array}{|c|c|c|c|}
\hline
     n-1\text{ voters of type 1} & B_{1} \ldots B_{i}  & A_{1}  \ldots A_{j-i}& C_{1} \ldots C_{m-j}  \\
\hline
     1\text{ voter of type 2} &  B_{1} \ldots B_{i}  & A_{j-i}  \ldots A_{1}& C_{1} \ldots C_{m-j}  \\
     \hline
\end{array}
\]

In $\alpha_{j}$ all $A$ and $B$ candidates are tied by score, and $A_{1}$ wins the tie. The voter of type 2 can swap $A_{1}$ for $C_{1}$. This will lower the score of $A_1$ from $n$ to $n-1$, while the other $A$ and $B$ candidates still have $n$. If $j-i>1$, the winner will be $A_2$. If $j-i=1$, the winner will be $B_1$. In either case, the outcome will be better for the manipulator.

    In $\alpha_{i}$ $B_1$ wins, so every voter gets his best outcome. No one has an incentive to manipulate.
\end{proof}

\begin{lemma}\label{lem:alphabigm}
For $n=2q, m\geq 2j-1$: $\alpha_{j} \nps \alpha_{i}$.
\end{lemma}
\begin{proof}
The profile consists of $j-1$ $B$ candidates, $j-1$ $C$ candidates, one $A$ candidate, and $m-2j+1$ $D$ candidates. The condition that $m\geq 2j-1$ guarantees that the top $j$ candidates of voters of types 1 and 2 intersect only at $A$:

\[
\begin{array}{|c|c|c|c|}
\hline
     q\text{ voters of type 1} & B_{1} \ldots B_{i} & B_{i+1} \ldots B_{j-1}A  &  C_1\dots C_{j-1}D_1\dots D_{m-2j+1}\\
\hline
     q-1\text{ voters of type 2} & AC_1\dots C_{i-1} & C_{i} \ldots C_{j-1}  &  B_1\dots B_{j-1}D_1\dots D_{m-2j+1}\\
     \hline
     1\text{ voter of type 3} & C_{1} \ldots C_{i} & C_{i+1} \ldots C_{j-1}A  & B_1\dots B_{j-1}D_1\dots D_{m-2j+1}\\
     \hline
\end{array}
\]

 Under $\alpha_i$, $A$ has $q-1$ points, $B_1$ through $B_{i}$ and $C_1$ through $C_{i-1}$ have $q$, $C_i$ has 1. The winner is $B_1$. The voter of type 3 can manipulate by swapping $A$ for $C_1$. $A$ will have $q$ points and will beat $B_1$ in the tie.
  
  Under $\alpha_j$, $A$ has $n$ points and is the winner. Voters of type 2 get their best choice elected. A voter of type 1 would rather have a $B$ candidate win, but $A$ has a lead of at least one point on these. Moving $A$ below the $j$th position will only drop the score by one, and $A$ will win the tie. The voter of type 3 would rather see a $C$ candidate win, but $A$ has a lead of at least one. Dropping $A$'s score will at most force a tie, which $A$ will win.
\end{proof}

\begin{lemma}\label{aoddsmall}
For $n=2q+1, m\geq 2j-1$: 
if  $i\geq 2$, then $\alpha_{j} \nps \alpha_{i}$.
\end{lemma}
\begin{proof}
The profile consists of $j$ $B$ candidates, $j-1$ $A$ candidates, and $m-2j+1$ $C$ candidates. The condition that $m\geq 2j-1$ guarantees that the top $j$ candidates of voters of types 1 and 2 intersect only at $B_1$:
\[
\begin{array}{|c|c|c|c|}
\hline
     q\text{ voters of type 1} & \overbrace{A_{1} \ldots A_{i}}^{\geq 2} &  B_{1}A_{i+1}\dots A_{j-1}&  B_2\dots B_jC_1\dots C_{m-2j+1} \\
\hline
     q\text{ voters of type 2} & B_{i} \ldots B_{1} &B_{i+1}\ldots  B_{j}  &  A_1\dots A_{j-1}C_1\dots C_{m-2j+1} \\
     \hline
     1\text{ voter of type 3} & B_{1} \ldots B_{i} & B_{i+1}\ldots  B_{j}  &A_1\dots A_{j-1} C_1\dots C_{m-2j+1} \\
     \hline
\end{array}
\]

Under $\alpha_i$, the winner is $B_1$ with $q+1$ points. An voter of type 2 can swap $B_1$ with $B_j$, and $B_2$ will win (since $i\geq 2$, $B_1$ and $B_2$ are necessarily distinct).

Under $\alpha_j$, $B_1$ wins with $n$ points. The other $B$ candidates have at most $n-1$ points, and the $A$ candidates have at most $n-2$. A voter of type 2 would rather see another $B$ candidate win, but such a voter can only lower the score of $B_1$ by 1, and $B_1$ will win the tie against any $B$ candidate. A voter of type 1 would rather see an $A$ candidate win, but $B_1$ has a lead of at least two points, so will beat the $A$ candidates by points even if ranked last.
\end{proof}

\begin{lemma}\label{lem:aione}
For all $n,m$: if $i=1$,  then $\alpha_{j} \nps \alpha_{i}$.
\end{lemma}
\begin{proof}

The profiles consist of candidates $A,B,C$, and $D_1$ through $D_{m-3}$.

Case one: $n$ even, $n=2q$.

\[
\begin{array}{|c|c|c|c|}
\hline
     q\text{ voters of type 1} &  B&AD_1\dots D_{j-2}&C  D_{j-1}\dots D_{m-3} \\
\hline
     q-1\text{ voters of type 2} & A&BD_1\dots D_{j-2}  &CD_{j-1}\dots D_{m-3}   \\
     \hline
     1\text{ voter of type 3} & C &AD_1\dots D_{j-2}&BD_{j-1}\dots D_{m-3}  \\
     \hline
\end{array}
\]

Under $\alpha_i$, $B$ has $q$ points, $A$ has $q-1$, and $C$ has 1. If $q\geq 2$, $B$ wins by score, and if $q=1$, $B$ wins the tie. The voter of type 3 can swap $C$ with $A$ to force a tie, which $A$ will win.

Under $\alpha_j$, $A$ is the winner with $n$ points. Voters of type 2 have no incentive to manipulate. A voters of type 1 would rather see $B$ win, but $B$ has at most $n-1$ points, so a voter of this type can at most force a tie, which $A$ will win. Likewise, the voter of type 3 would rather $C$ win, who has 1 point. Decreasing $A$'s score will at most force a tie.

Case two: $n$ odd, $n=2q+1$.

\[
\begin{array}{|c|c|c|c|}
\hline
     q\text{ voters of type 1} &  B&CD_1\dots D_{j-2}& A D_{j-1}\dots D_{m-3} \\
\hline
     q\text{ voters of type 2} & A&BD_1\dots D_{j-2}  &CD_{j-1}\dots D_{m-3}   \\
     \hline
     1\text{ voter of type 3} & C &BD_1\dots D_{j-2}   &AD_{j-1}\dots D_{m-3}  \\
     \hline
\end{array}
\]

Under $\alpha_i$, $A$ wins by tie-breaking. The voter of type 3 can swap $C$ with $B$ to give $B$ a points victory.

 Under $\alpha_j$, $B$ wins with $n$ points. A voter of type 2 would rather see $A$ win, but $B$ has at least a two point lead, so he cannot force a tie. The voter of type 3 would rather $C$ win, but $C$ is behind by at least one point, so he can at best force a tie, which $B$ will win.

\end{proof}

\begin{corollary}\label{lem:ambig}
For all $n,m \geq 2j$: $\alpha_{j} \nps \alpha_{i}$.
\end{corollary}
\begin{proof} \autoref{lem:alphabigm}, \autoref{aoddsmall}, and \autoref{lem:aione}.
\end{proof}

\begin{lemma}\label{lem:ammiddle}
For all $n, m<2j$: if $i\geq 2$, then $\alpha_{j} \nps \alpha_{i}$.
\end{lemma}

\begin{proof}

Case one: $m\geq i+j$. 

Since $m\geq i+j = 2i + (j-i)$, we can guarantee the existence of $2i$ $A$ candidates and $(j-i)$ $B$ candidates. The remaining $m-(i+j)$ candidates are the $C$ candidates. Observe that since $m < 2j$, the number of the $C$ candidates is smaller than the number of the $B$ candidates $(m-j-i < j-i)$.

\[
\begin{array}{|c|c|c|c|}
\hline
     \floor{n/2}\text{ voters of type 1} & \overbrace{A_{1} \ldots A_{i}}^{\geq2}&B_1\dots B_{j-i}&A_{2i}\dots A_{i+1}\overbrace{C_1\dots C_{m-(i+j)}}^{<j-i} \\
\hline
     \ceil{n/2}\text{ voter of type 2} & A_{i+1} \dots A_{2i}& B_{1} \ldots B_{j-i} & A_{i}\dots A_1 C_1\dots C_{m-(i+j)}\\
     \hline
\end{array}
\]

Under $\alpha_i$, if $n$ is even the winner is $A_1$. A voter of type 2 can manipulate by swapping $A_{2i}$ with $A_{i}$, giving $A_i$ $n/2+1$ points. Since $i\geq 2$, $A_i\neq A_1$. If $n$ is odd, the winner is $A_{i+1}$ with $\ceil{n/2}$ points. A voter of type 1 can swap $A_1$ with $A_{2i}$ to give $A_{2i}$ $\ceil{n/2}+1$ points. Since $i\geq 2$, $A_{2i}\neq A_{i+1}$.

Under $\alpha_j$, all the $B$ candidates have $n$ points, and the winner is $B_1$. A voter of type 1 would rather see one of $A_1,\dots, A_i$ win. Since he cannot raise the score of these, he will have to lower the score of the $B$ candidates. However, there are more $B$ candidates than $C$ candidates, so if the voter were to rank all the $B$ candidates below the $j$th position, he would necessarily raise the score of one of $A_{i+1},\dots,A_{2i}$ to $\ceil{n/2}+1$. That candidate would then win by score, and he is even worse for the manipulator than $B_1$.


Likewise, a voter of type 2 would rather see one of $A_{i+1},\dots, A_{2i}$ win. He can attempt to rank all the $B$ candidates below $j$, but then one of $A_1,\dots,A_i$ will get $\floor{n/2}+1\geq\ceil{n/2}$ points and win the election (possibly by tie-breaking).

 Case two: $m< i+j$.
 
 In the profile below we have $i$ $C$ candidates, $j-i$ $B$ candidates, and $m-j$ $A$ candidates. Since $m<i+j$, $m-j<i$, so the voters of type 1 can rank all the $A$ candidates in the top $i$ positions, as well as at least one $B$ candidate.
  \[
\begin{array}{|c|c|c|c|}
\hline
     n-1\text{ voters of type 1} & \overbrace{B_1\dots B_{i-(m-j)}}^{\geq1}\overbrace{A_{1} \ldots A_{m-j}}^{\geq1}&\text{any order} &\text{any order} \\
\hline
     1\text{ voter of type 2} & C_{1} \dots C_{i}& B_{1} \ldots B_{j-i} & A_{1}\dots A_{m-j}\\
     \hline
\end{array}
\]

Under $\alpha_i$, the winner is $A_1$. The voter of type 2 can manipulate by ranking $B_1$ first.

Under $\alpha_j$, the winner is $B_1$ (either by score, or winning a tie against a $C$ candidate). The voters of type 1 get their best choice elected. The voter of type 2 would rather see a $C$ candidate win, but to do so he would have to lower the score of the $B$ candidates. If he ranks any $B$ candidate below the $j$th position, he would have to rank one of the $A$ candidate above -- that candidate would then win the election with $n$ points, and the outcome would be worse than $B_1$.
 \end{proof}

\begin{corollary}\label{lem:amsmall}
For all $n, m < 2j$: $\alpha_{j} \nps \alpha_{i}$.
\end{corollary}
\begin{proof}
\autoref{lem:aione} and \autoref{lem:ammiddle}.
\end{proof}

\begin{theorem}\label{thm:approval}
For all $n, m$: $\alpha_{i} \times_{PS} \alpha_{j}$.
\end{theorem}
\begin{proof}
By \autoref{lemLaij}, $\alpha_i\nps\alpha_j$. By \autoref{lem:ambig} and \autoref{lem:amsmall}, $\alpha_j\nps\alpha_i$.
\end{proof}

\section{The $k$-Borda Family}

As before, we fix $i<j$. In this section we will show that for $n=2,j\neq m-1$, $\beta_j\pss\beta_i$ (\autoref{cor:main}), but in all other cases the rules are incomparable (\autoref{thm:bordainc}, \autoref{cor:betaibetaj}, \autoref{prop:betakborda}).

We make use of a standard result about the manipulability of scoring rules:

\begin{lemma}\label{lem:normalform}
Consider a profile $P$, and scoring rule $f$. Let $w$ be the winner under sincere voting, $f(P)=w$. Call all the candidates voter $i$ perceives to be at least as bad as $w$ (including $w$) the bad candidates. The others, the good candidates. Order the good candidates $g_1,\dots, g_q$ and the bad candidates $b_1,\dots,b_r$ from the highest to the lowest scoring in $P_{-i}$. In case of equal scores, order candidates by their order in the tie-breaking. We claim that if $i$ can manipulate $f$ at $P$, he can manipulate with the following vote:
$$P_i^*=g_1\succ\dots\succ g_q\succ b_r\succ\dots\succ b_1.$$
\end{lemma}
\begin{proof}
Let $\score(c,P)$ be the score of candidate $c$ at profile $P$. Suppose voter $i$ can manipulate at $P$. That is, there is a $P_i'$ such that $f(P_i',P_{-i})=g_j$. In order to be the winner, $g_j$ must have the highest score.
\begin{equation}\label{eqngiwins}
    \score (g_j, P_{i}', P_{-i}) \geq \max_{c\neq g_j}(\score ( c, P_{i}', P_{-i})).
\end{equation}
Observe that $\score (g_j, P_{i}', P_{-i})=\score (g_j, P_{-i})+s_k$, where $k$ is the position in which $g_j$ is ranked in $P_i'$. By ranking $g_1$ first in $P_i^*$ it follows that $\score (g_1, P_{i}^{*}, P_{-i})=\score (g_1, P_{-i})+s_1$, and observe that $\score (g_1, P_{-i})\geq \score (g_j, P_{-i})$, and $s_1\geq s_k$. Thus:
\begin{equation}\label{eqng1hasmore}
    \score ( g_1, P_{i}^{*}, P_{-i})\geq \score (g_j, P_{i}', P_{-i}).
\end{equation}

We now claim that the score of the highest scoring bad candidate in $(P_i',P_{-i})$ is no higher than in $(P_i^*,P_{-i})$. For contradiction, suppose that $b_p$ is the highest scoring bad candidate in $(P_i^*,P_{-i})$, and his score is higher than any bad candidate in $(P_i',P_{-i})$. Observe that $\score ( b_p, P_{i}^{*}, P_{-i}) = \score ( b_p, P_{-i}) + s_{m-p+1}$. Since $\score ( b_1, P_{-i}),\dots, \score ( b_p, P_{-i})\geq \score ( b_p, P_{-i})$, this means that bad candidates $b_1,\dots,b_p$ must all get strictly less than $s_{m-p+1}$ points in $P_i'$. However there are $p$ such candidates, and only $p-1$ positions below $m-p+1$.

Since the highest scoring candidate in $(P_i',P_{-i})$ has at least as many points as the highest scoring bad candidate, it follows that:
\begin{equation}\label{eqnbjhasless}
    \max_{c\neq g_j}(\score ( c, P_{i}', P_{-i}))\geq \max_{b\in\set{b_1,\dots,b_r}}(\score ( b, P_{i}^{*}, P_{-i})).
\end{equation}
Combining \ref{eqngiwins}, \ref{eqng1hasmore}, and \ref{eqnbjhasless}, we conclude that $g_1$ is among the highest scoring candidates in $(P_i^*,P_{-i})$. If at least one of the inequalities is strict, $g_1$ has more points than any bad candidate and we are done.

Suppose then that all the inequalities are equal. Observe that this implies that if $g_1\neq g_j$, then $g_1$ must come before $g_j$ in the tie-breaking order. To see this, observe that if we assume $\score ( g_1, P_{i}^{*}, P_{-i})= \score (g_j, P_{i}', P_{-i})$, then it follows that $\score ( g_1,  P_{-i}) + s_1= \score (g_j,  P_{-i}) + s_k$, where $k$ is the position in which $g_j$ is ranked in $P_i'$. Since $\score (g_1, P_{-i})\geq \score (g_j, P_{-i})$, and $s_1\geq s_k$, the only way this is possible is if $\score (g_1, P_{-i})= \score (g_j, P_{-i})$. By definition, in the case of equal scores in $P_{-i}$, the candidate that is labelled $g_1$ must have priority in the tie-breaking.

If $g_1$ also wins the tie against any bad candidate, we are done. For contradiction, suppose a bad candidate $b_p$ wins the tie given $P_i^*$. This means that $b_p$ beats $g_1$ and $g_j$ in the tie-breaking. Observe that $b_p$ is ranked in position ${m-p+1}$ in $P_i^*$. Since $b_p$ loses in $(P_i',P_{-i})$, $b_p$ must have been ranked lower than ${m-p+1}$ in $P_i'$. This means $b_p$ was ranked lower than at least $m-p+1$ candidates. Since $m=q+r$, and there are $q$ candidates, this means $b_p$ was ranked lower than at least $r-p+1$ bad candidates. In $P_i^*$, $b_p$ is ranked below exactly $r-p$ bad candidates, so there must exist a bad candidate that was ranked above $b_p$ in $P_i'$, but is ranked below $b_p$ in $P_i^*$. Call this candidate $b_t$. By definition of $P_i^*$, it must be the case that $\score(b_t,P_{-i}) >\score(b_p,P_{-i})$ or $\score(b_t,P_{-i}) =\score(b_p,P_{-i})$ and $b_t$ wins the tie. But that is impossible, because then $b_t$ would have gained at least as many points in $(P_i',P_{-i})$ as $b_p$ did in $(P_i*,P_{-i})$, and since the score of $b_p$ in $(P_i*,P_{-i})$ is equal to $g_1$, it means $b_t$ has at least as many points in $(P_i',P_{-i})$ as $g_j$, so wins either by points or by tie-breaking.
\end{proof}

\begin{lemma}\label{prop:2qij}
For $n=2q$, all $m$: $\beta_i\nps\beta_{j}$.
\end{lemma}
\begin{proof}
Consider a profile with one $A$ candidate, one $B$ candidate, and $m-2$ $C$ candidates:

\[
\begin{array}{|c|c|c|c|}
\hline
     1\text{ voter of type 1} & A  C_{1} \ldots C_{i-1}&  C_{i} \dots C_{j-2} B&   C_{j-1}\dots C_{m-2}\\
\hline
    1\text{ voter of type 2} & B A C_{m-2} \ldots C_{m-i+1}& C_{m-i}\dots C_{m-j+1} &  C_{m-j}\dots C_1\\
     \hline
     q-1\text{ voters of type 3} & A  C_{1} \ldots C_{i-1} & C_i\dots C_{j-1} &   C_{j-2}\dots C_{m-2}B\\
\hline
    q-1\text{ voter of type 4} & B  C_{m-2} \ldots C_{m-i}& C_{m-i-1}\dots C_{m-j} &  C_{m-j-1}\dots C_1A\\
     \hline
\end{array}
\]
Under $\beta_{j}$, $A$ has $qj+(j-1)$ points. $B$ has $qj+1$. $C_1$ is the highest scoring $C$ candidate with $q(j-1)$. The winner is $A$. However, the voter of type 2 can rank $A$ last and shift the $C$ candidates up one. This gives $A$ a score of $qj$, $B$'s score is still $qj+1$, and a $C$ candidate's is at most $q(j-1)+1$. $B$ wins a points victory.

 Under $\beta_i$, $A$ has $qi +(i-1)=qi+i-1$ points. $B$ has $qi$. $C_1$ has $q(i-1)$, the other $C$ candidates no more. Voters of type 1 and three have no incentive to manipulate. The voter of type 2 would rather see $B$ win, but by \autoref{lem:normalform} this would mean ranking $A$ last, and $A$ would still have $qi$ points and win the tie. A voter of type 4 would rather see anyone win, and by \autoref{lem:normalform} this involves putting either $B$ or $C_1$ first. $B$ is already ranked first and does not win, and putting $C_1$ first would give $C_1$ $q(i-1)+i=qi-q+i$ points, which is less than $A$.
\end{proof}

\begin{lemma}\label{prop:oddij}
For $n=2q+1$, all $m$: $\beta_i\nps\beta_{j}$.
\end{lemma}

\begin{proof}
Consider a profile with one $A$ candidate, one $B$ candidate, and $m-2$ $C$ candidates:
\[
\begin{array}{|c|c|c|c|}
\hline
    1\text{ voter of type 1} & BC_1\dots C_{i-1}&C_{i}\dots C_{j-2}A&C_{j-1}\dots  C_{m-2}\\
     \hline
     q\text{ voters of type 2} &BAC_1\dots C_{i-2}&C_{i-1}\dots C_{j-2}&C_{j-1}\dots C_{m-2}\\
     \hline
q\text{ voters of type 3} &AB C_{m-2}\dots C_{m-i+1}& C_{m-i}\dots C_{m-j+1}&C_{m-j}\dots C_1\\
     \hline
\end{array}
\]
Under $\beta_i$, $A$ has $qi+q(i-1)=2qi-q$ points. $B$ has $(q+1)i+q(i-1)=2qi-q+i$. All the $C$ candidates are Pareto dominated by $B$, so the winner is $B$. A voter of type 3 would like to see $A$ win, but if he ranks $B$ last, $A$ will have $2qi-q$ points to $B$'s $2qi-q+i-(i-1)=2qi-q+1$, so $B$ would still win.

Under $\beta_j$, $A$ has $qj+q(j-1)+1=2qj-q+1$, $B$ has $(q+1)j+q(j-1)=2qj-q+j$. If a voter of type three ranks $B$ last and shifts the $C$ candidates up one, $B$ will have $2qj-q+j-(j-1)=2qj-q+1$, tying with $A$, and $A$ wins the tie. It remains to check that $A$ will have more points than the highest scoring $C$ candidate, which is clearly $C_1$. After the manipulation, $C_1$'s score will increase by at most one. $C_1$'s score before manipulation is $j-1+q(j-2)$, so $A$ will beat $C_1$ if:
\begin{align*}
    2qj-q+1&\geq j  + qj-2q,\\
    qj-q&\geq j -1 -2q,\\
    qj-j&\geq -1 -q,\\
    (q-1)j&\geq -1 -q.
\end{align*}
Which is always satisfied.
\end{proof}

\begin{corollary}\label{cor:betaibetaj}
For all $n, m$: $\beta_i\nps\beta_{j}$.
\end{corollary}
\begin{proof}
\autoref{prop:2qij} and \autoref{prop:oddij}.
\end{proof}

\begin{lemma}\label{prop:betajibordaodd}
For $n=2q+1$, all $m$: $\beta_{j}\nps \beta_i$.
\end{lemma}
\begin{proof}
Consider the following profile:

\[
\begin{array}{|c|c|c|c|c|c|}
\hline
   q\text{ voters of type 1} & ABC_{m-2}\dots C_{m-i+1}&C_{m-i}\dots C_{m-j+1}& C_{m-j}\dots C_1 \\
    \hline
    q-1\text{ voters of type 2} & BAC_1\dots C_{i-2}& C_{i-1}\dots C_{j-2}&C_{j-1}\dots  C_{m-2} \\
    \hline
    1\text{ voters of type 3} & BC_1\dots C_{i-1}& C_i\dots C_{j-1}&C_j\dots  C_{m-2} A\\
\hline
1\text{ voter of type 4} & C_1\dots C_i& B C_{i+1} \dots C_{j-1}&C_{j}\dots  C_{m-2} A \\
\hline
\end{array}  
\]

Under $\beta_i$, $A$ has $qi+(q-1)(i-1)=2qi-q-i+1$ points. $B$ has $qi+q(i-1)=2qi-q$. $C_1$ is clearly the highest scoring $C$ candidate, and has $i+(i-1)+(q-1)(i-2)=qi+i-2q+1$ points. If $i>1$, $B$ wins a points victory, but a voter of type 1 can rank $B$ last to force a tie, which $A$ will win. If $i=1$, then $A$ wins by tie-breaking, but the voter of type 4 can rank $B$ first to make $B$ the winner.

Under $\beta_j$, $A$ has $qj+(q-1)(j-1)$. $B$ has $qj+q(j-1)+(j-i)$. A voter of type 1 would like to manipulate in favour of $A$, but if he ranks $B$ last, $B$'s score will only drop by $j-1$, and $B$ will still win.

The voter of type 4 would rather see one of $C_1$ through $C_i$ win. Observe that an upper bound on the score a $C$ candidate can get from the voters of type 1 through 3 is $q(j-1)$ -- for $q-1$ voters of types 1 and 2, each time the first group gives the candidate $j-2-k$ points, the other gives at most $k$ points, for an upper bound of $(q-1)(j-2)$. The voter of type 3 ranks all $C$ candidates one position higher, so combined with the remaining voter of type 1 the contribution to the candidate's score is at most $j-1$, which gives a total of $(q-1)(j-2)+(j-1) <q(j-1)$. If the voter ranks $B$ last and the $C$ candidate first, $B$ will still have $qj+q(j-1)$ points to the $C$'s candidate $j+q(j-1)$, so $B$ will still win.
\end{proof}

\begin{lemma}\label{prop:betajiqbigbordaeven}
For $n=2q$, all $m$: if $q>2$, $\beta_{j}\nps \beta_i$.
\end{lemma}
\begin{proof}
Consider the following profile:
\[
\begin{array}{|c|c|c|c|c|c|}
\hline
   q-1\text{ voters of type 1} & ABC_{m-2}\dots C_{m-i+1}&  C_{m-i}  \dots C_{m-j+1}& C_{m-j}\dots C_1 \\
\hline
    q-2\text{ voters of type 2} & BA C_{1}\dots C_{i-2}& C_{i-1}  \dots C_{j-2}& C_{j-1}\dots C_{m-2}  \\
\hline
1\text{ voter of type 3} & BC_{1}\dots C_{i-1}& C_{i}  \dots C_{j-1}& C_{j}\dots C_{m-2} A \\
\hline
1\text{ voter of type 4} & C_{1}\dots C_{i}&B C_{i+1}  \dots C_{j-1}& C_{j}\dots C_{m-2} A \\
\hline
1\text{ voter of type 5} & C_{m-2}\dots C_{m-i-1}&B C_{m-i-2}  \dots C_{m-j}& C_{m-j-1}\dots C_{1} A \\
\hline
\end{array}  
\]

Case one: $m>3$, and hence $C_1\neq C_{m-2}$.

Under $\beta_i$, $A$ has $(q-1)i+(q-2)(i-1)$ points and $B$ has $(q-1)i+(q-1)(i-1)$. A $C$ candidate has at most $(q-2)(i-2)+(i-1)+(i+1)$ -- observe that if the candidate gets $i-2-k$ points from a voter of type 1, he gets at most $k$ from a voter of type 2, which gives us at most $(q-2)(i-2)$ from $q-2$ of each type of voter; the voter of type 3 gives one more point to the candidate, so paired with the remaining voter of type 1, the contribution is $i-1$; as for the voters of voters of type 4 and 5, if one gives the candidate $i-k$ points, the other gives at most $k+1$, for the remaining $(i+1)$.

Since $q>2$, $A$ and $B$ have more points than the $C$ candidates. If $i>1$, $B$ also beats $A$ by points, but a voter of type 1 can rank $B$ last and shift the $C$ candidates up one to force a tie between $A$ and $B$. This operation will raise the score of a $C$ candidate by at most 1, so such a candidate will at worst enter the tie, which $A$ wins. If $i=1$, $A$ wins by tie-breaking, but a voter of type 4 can rank $B$ first to give him one more point.

Under $\beta_j$, $A$ has $(q-1)j+(q-2)(j-1)$ points, $B$ has $(q-1)j+(q-1)(j-1)+2(j-i)$. $B$ has more points than $A$, and a voter of type 1 can no longer change this by ranking $B$ last.

A voter of type 4 would like to manipulate in favour of one of $C_1$ through $C_i$. As we have argued above, such a candidate would get no more than $(q-2)(j-2)+(j-1)$ from voters of types 1 through 3, which we round up to $(q-1)(j-1)$. If the manipulator ranks this candidate first, he will get at most $2j$ from the voters of type $4,5$ for an upper bound of $2j+(q-1)(j-1)$. In comparison, $B$ gets $(q-1)j+(q-1)(j-1)$ from the voters of type 1 through 3. Since $q>2$, $B$ would still win. The argument for the voter of type 5 is analogous.

Case two: $m=3$. The same profile we had above collapses to the following:

\[
\begin{array}{|c|c|c|c|c|c|}
\hline
   q-1\text{ voters of type 1} & A&B&C \\
\hline
    q-2\text{ voters of type 2} & B&A& C\\
\hline
1\text{ voter of type 3} & B&C& A \\
\hline
1\text{ voter of type 4} & C&B& A \\
\hline
1\text{ voter of type 5} & C&B&A \\
\hline
\end{array}  
\]
The argument with respect to $A$ and $B$ is unchanged. We need only verify that $C$ cannot win under $\beta_i$ or $\beta_j$.

Under $\beta_i$, $C$ has exactly 2 points. $A$ and $B$ are tied with $q-1$, and the voter of type 4 can only manipulate in favour of $B$ by ranking $B$ first.

Under $\beta_j$, $C$ has exactly 5 points. $B$ has at least 8, and a voter of type 4 or 5 can only lower $B$'s score by one point.
\end{proof}

\begin{lemma}\label{prop:betajiqbigbordafour}
For $n=4$, all $m$: $\beta_{j}\nps \beta_i$.
\end{lemma}
\begin{proof}
Case one: $i>1$.
\[
\begin{array}{|c|c|c|c|c|c|}
\hline
   2\text{ voters of type 1} & BC_{m-2}\dots C_{m-i}&  C_{m-i-1}  \dots C_{m-j}& C_{m-j-1}\dots C_1A \\
\hline
1\text{ voter of type 2} & ABC_{1}\dots C_{i-2}& C_{i-1}  \dots C_{j-2}& C_{j-1}\dots C_{m-2} \\
\hline
1\text{ voter of type 3} & AC_{1}\dots C_{i-1}& C_{i}  \dots C_{j-2}B& C_{j-1}\dots C_{m-2} \\
\hline
\end{array}  
\]

Under $\beta_i$, $B$ wins with $2+i-1$ points. The voter of type 2 can manipulate by ranking $B$ last.

Under $\beta_j$, $B$ wins with $3j$ points. If the voter of type 2 ranks $B$ last, $B$ will still have $2j+1$, beating $A$. The voter of type 3, likewise, cannot manipulate in favour of $A$, but could try to manipulate in favour of a $C$ candidate. If he ranks $B$ last and $C_{i-1}$ first, then $B$ will have a score of $3j-1$. We can bound $C_{i-1}$'s score by $j$ (from one voter of type 1 and type 2) $+j$ (the manipulator ranks $C_{i-1}$ first) $+x$ (the points from the remaining voter of type 1). In order for $C_{i-1}$ to win, we must have $2j+x > 3j-1$, which is clearly impossible.

Case two: $i=1$.
\[
\begin{array}{|c|c|c|c|c|c|}
\hline
   2\text{ voters of type 1} & B&AC_{m-2}  \dots C_{m-j+1}& C_{m-j}\dots C_1 \\
\hline
1\text{ voter of type 2} & A&C_{1}\dots C_{j-1}& C_{j-1}\dots C_{m-2}B \\
\hline
1\text{ voter of type 3} & C_{1}&AC_2\dots C_{j-1}& C_{j}\dots C_{m-2}B \\
\hline
\end{array}  
\]
Under $\beta_1$, $B$ is the winner, but the voter of type 3 can manipulate in favour of $A$.

Under $\beta_j$, $A$ has $4j-3$ points to $B$'s $2j$. Since $j\geq 2$, $A$ is the winner. A voter of type 1 would rather see $B$ win, but even if he ranks $A$ last, $A$ will still have $3j-2\geq 2j$ points. A voter of type 3 would rather see $C_1$ win, but $C_1$ has $2j-1$ points, so would lose to $B$ no matter what the voter does.
\end{proof}

\begin{corollary}\label{corbetaji}
For $n>2$, all $m$: $\beta_{j}\nps \beta_i$.
\end{corollary}
\begin{proof}
\autoref{prop:betajibordaodd}, \autoref{prop:betajiqbigbordaeven}, and \autoref{prop:betajiqbigbordafour}.
\end{proof}

\begin{theorem}\label{thm:bordainc}
For $n>2$, all $m$: $\beta_i\xps\beta_j$.
\end{theorem}
\begin{proof}
\autoref{cor:betaibetaj} and \autoref{corbetaji}.
\end{proof}

Thus far the story resembles that of $k$-approval. However, in the case of $n=2$, a hierarchy of manipulability is observed:

\begin{theorem}\label{thm:bordan2}
For $n=2$, $m>k+2$: $\beta_{k+1} \ps \beta_{k}$.
\end{theorem}
\begin{proof}
Let voter 1's preferences be $c_1\succ_1\dots\succ_1c_m$ and voter 2's $b_1\succ_2\dots\succ_2b_m$. Note that $c_1\neq b_1$, else manipulation would not be possible.

Let $\beta_k(P_1,P_2)=d$ and $\beta_{k+1}(P_1,P_2)=e$.  We consider whether or not $d=e$ by cases.

Case one: $d=e$.

Assume voter 1 can manipulate $\beta_k$ in favour of $c_q\succ_1 d$. By \autoref{lem:normalform}, this means $c_q$ is the winner in the following profile:

\[
\begin{array}{|c|c|c|c|}
\hline
     P_1^* & c_q b_{m}\dots b_{m-k+1}&b_{m-k} &b_{m-k-1}\dots b_1 \\
\hline
     P_2 & b_1  \ldots b_k& b_{k+1} & b_{k+1}\dots b_m \\
\hline
\end{array}  
\]

Let us consider who the winner must be under $\beta_{k+1}(P_1^*,P_2)$. Observe that the score of a candidate under $\beta_{k+1}$ is at most two points higher than under $\beta_k$ -- it will increase by one point for each voter who ranks the candidate in the top $k+1$ positions.

If $c_q$'s points increase by 2 points then we are done -- whenever $c_q$ has more points than $f$ under $\beta_k$, $c_q$ still has more points under $\beta_{k+1}$; and if $c_q$ is tied with $f$ under $\beta_k$ then that must mean $c_q$ beats $f$ in the tie, under $\beta_{k+1}$ $c_q$ will either tie with $f$ and win the tie, or have more points outright. Thus voter 1 can manipulate in favour of $c_q\succ_1 d$.

If $c_q$'s points increase by 1, then that must mean that voter 2 does not rank $c_q$ in the top $k+1$ positions, and a fortiori in the top $k$ positions. Thus under $\beta_k(P_1^*,P_2)$ $c_q$ has $k$ points, $b_1$ has $k$ points, and the other candidates strictly less. Under $\beta_{k+1}$ $c_q$ and $b_1$ will still tie at $k+1$, and, since $m>k+2$, the other candidates will still have strictly less.

Case two: $d\neq e$.

As we have argued before, the score of a candidate can increase by at most two points when going from $\beta_k$ to $\beta_{k+1}$. Since $d$ wins under $\beta_k$ but $e$ wins under $\beta_{k+1}$, this means that $e$'s score must increase by 2 and $d$'s by 1. This means that one voter does not rank $d$ in the top $k+1$ positions, and, since $d$ must still win under $\beta_k$, this means the other voter must rank $d$ first (at least one candidate will have a score of $k$, so the winner's score must be at least $k$). Since voter 1 is the one with an incentive to manipulate, this means the sincere profile must be the following: 

\[
\begin{array}{|c|c|c|c|}
\hline
     \text{Voter }1 & c_1 \ldots c_{i-1} e c_{i+1} \dots c_m \\
\hline
     \text{Voter }2 & d b_2 \ldots b_{j-1} eb_{j+1} \dots b_m \\
\hline
\end{array}  
\]

Under $\beta_k$ $d$ either has one point more than $e$, or they are tied and $d$ wins the tie. Under $\beta_{k+1}$ $e$ wins, which means $d$'s score increases by 1 and $e$'s by two -- thus $d$ cannot be in the top $k+1$ positions of voter 1. But this means in the sincere profile both $d$ and $c_1$ have $k$ points under $\beta_k$, and $d$ wins the tie. Voter 2 can thus manipulate $\beta_{k+1}$ as follows:

\[
\begin{array}{|c|c|c|c|}
\hline
     \text{Voter }1 & c_1 \ldots c_{i-1} e c_{i+1} \dots c_m \\
\hline
     \text{Voter }2 & d c_m \dots c_1 \\
\hline
\end{array}  
\]

Both $d$ and $c_1$ have $k+1$ points, since $m>k+2$ the other candidates have strictly less, and $d$ wins the tie.
\end{proof}

\begin{corollary}\label{cor:main}
For $n=2$, $k<m-2$: $\beta_{j}\pss\beta_i$.
\end{corollary}
\begin{proof}
Transitivity of $\ps$, \autoref{thm:bordan2}, and \autoref{cor:betaibetaj}.
\end{proof}

To finish, we observe that the proviso that $m>k+2$ really is necessary -- the Borda rule proper ($\beta_{m-1}$) is incomparable with $\beta_k$.

\begin{proposition}\label{prop:betakborda}
For $n=2$, all $m$: $\beta_{m-1}\nps \beta_k$.
\end{proposition}
\begin{proof}
Consider the following profile, with a $B$ candidate, a $C$ candidate, and $m-2$ $A$ candidates:

\[
\begin{array}{|c|c|c|c|}
\hline
   \text{Voter }1 & B C A_{1}\dots A_{k-2}&A_{k-1} \dots A_{m-2}\\
\hline
    \text{Voter }2 & C A_1 \dots A_{k-1}&A_{k}\dots A_{m-2} B\\
\hline
\end{array}  
\]

    Under $\beta_k$, if $k>1$ then $C$ is the winner with $2k-1$ points. Voter one can manipulate by voting $B\succ A_m\succ\dots\succ A_1\succ C$. This way $B$ will have $k$ points, and the other candidates strictly less. If $k=1$ then $B$ is the winner by tie-breaking. Voter 2 can manipulate by voting for $A_1$.

    Under $\beta_{m-1}$, $C$ is the winner. Voter 2 has no incentive to manipulate, voter 1 would rather see $B$ win. By \autoref{lem:normalform}, if this is possible then it is possible in the following profile:
    
    \[
\begin{array}{|c|c|c|c|}
\hline
   \text{Voter }1 & B A_{m-2} \dots A_1 C \\
\hline
    \text{Voter }2 & C A_1 \dots  A_{m-2} B \\
\hline
\end{array}  
\]

However, in this profile all candidates are tied with $m-1$ points, and $A_1$ wins the tie, which is worse than $C$ for the manipulator.

\end{proof}

\section{Conclusion}

In this paper we have shown:
\begin{enumerate}
    \item For any choice of $n,m$: $\alpha_i\nps\alpha_j$;
    \item For $n=2,i<j,j\neq m=1$: $\beta_j\pss\beta_i$;
    \item In every other instance, $\beta_i\nps\beta_j$.
\end{enumerate}

These results are negative in nature. Even in the case of two natural, hierarchical families of scoring rules, the notion of Pathak and S\"onmez fails to make a meaningful distinction between their manipulability. The quest for a useful framework for comparing the manipulability of voting rules continues.


\begin{thebibliography}{}  

\bibitem[Aleskerov and Kurbanov, 1999]{Aleskerov1999}
Aleskerov, F.\ and E. Kurbanov (1999). {\em  Degree of manipulability of social choice procedures}. In: Current trends in economics (Alkan A., C.\,D,\ Aliprantis, and N.\,C.\ Yannelis, eds), vol.\,8, pp.\,595--609. Springer: Berlin, Heidelberg.

\bibitem[2016]{Arribillaga2016}
Arribillaga, R.\,P.\ and J.\ Mass{\'o} (2016). {\em  Comparing generalized median voter schemes according to their manipulability}. Theoretical Economics, {\bf 11(2)}, 547--586.

\bibitem[2017]{Arribillaga2017}
Arribillaga, R.\,P.\ and J.\ Mass{\'o} (2017). {\em  Comparing Voting by Committees According to Their Manipulability}. American Economic Journal: Microeconomics, {\bf 9(4)}, 74--107.

\bibitem[(Barbera et al., 1991)]{Barbera1991}
Barber\`a, S., H. Sonnenschein, and L. Zhou (1991). {\em  Voting by Committees}. Econometrica, {\bf 59(3)}, 595--609.

\bibitem[Bartholdi et al., 1989]{Bartholdi1989}
Bartholdi, J.\,J, C.\,A.\ Tovey, and M.\,A.\ Trick (1989). {\em  The computational difficulty of manipulating an election}. Social choice and Welfare, {\bf 6(3)}, 227--241.

\bibitem[Bonkoungou and Nesterov, 2021]{Bonkoungou2021}
Bonkoungou, S.\ and A.\ Nesterov (2021). {\em  Comparing school choice and college admissions mechanisms by their strategic accessibility}. Theoretical Economics, {\bf 16(3)}, 881--909.

\bibitem[Conitzer et al., 2007]{Conitzer2007}
Conitzer, V., T.\ Sandholm, and J.\ Lang (2007). {\em  When are elections with few candidates hard to manipulate?}. Journal of the ACM, {\bf 54(3)}, 14--es.

\bibitem[Decerf and Van der Linden, 2021)]{Decerf2021}
Decerf, B.\ and M.\ Van der Linden (2021). {\em  Manipulability in school choice}. Journal of Economic Theory, {\bf 197}.

\bibitem[(Dumett and Farquharson, 1961)]{Dumett1961}
Dummett, M.\ and R. Farquharson (1961). {\em  Stability in Voting}. Econometrica, {\bf 29(1)}, 33--43.

\bibitem[(Elkind et al., 2017)]{Elkind2017}
Elkind, E., M.\ Lackner, and D.\ Peters (2017). {\em  Structured Preferences}. In: Trends in Computational Social Choice (Endriss, U., ed), pp.\,187--207. AI Access.

\bibitem[(Faliszewski et al., 2011)]{Faliszewski2011}
Faliszewski, P., E.\ Hemaspaandra, L.\,A.\ Hemaspaandra, and J.\ Rothe (2011). {\em  The shield that never was: Societies with single-peaked preferences are more open to manipulation and control}. Information and Computation, {\bf 209(2)}, 89--107.

\bibitem[Gibbard, 1973]{Gibbard1973}
Gibbard, A.\ (1973). {\em  Manipulation of Voting Schemes: A General Result}. Econometrica, {\bf 41(4)}, 587--601.

\bibitem[Kelly, 1993]{Kelly1993}
Kelly, J.\,S.\ (1993). {\em  Almost all social choice rules are highly manipulable, but a few aren't}. Social choice and Welfare, {\bf 10(2)}, 161--175.

\bibitem[(Moulin, 1980)]{Moulin1980}
Moulin, H.\ (1980). {\em  On Strategy-Proofness and Single Peakedness}. Public Choice, {\bf 35(4)}, 437--455.

\bibitem[Nitzan, 1985]{Nitzan1985}
Nitzan\`a, S.\ (1985). {\em  The vulnerability of point-voting schemes to preference variation and strategic manipulation}. Public choice, {\bf 47(2)}, 349--370.

\bibitem[Pathak and S\"onmez (2013)]{Pathak2013}
Pathak, P.\,A.\ and T. S\"onmez (2013). {\em  School admissions reform in Chicago and England:
Comparing mechanisms by their vulnerability to manipulation}. American Economic Review, {\bf 103(1)}, 80--106.

\bibitem[Satterthwaite, 1975]{Satterthwaite1975}
Satterthwaite, M.\,A.\ (1975). {\em  Strategy-proofness and Arrow's conditions: Existence and correspondence theorems for voting procedures and social welfare functions}. Journal of Economic Theory, {\bf 10(2)}, 187--217.

\bibitem[Walsh, 2011]{Walsh2011}
Walsh, T.\ (2011). {\em  Is computational complexity a barrier to manipulation?}. Annals of Mathematics and Artificial Intelligence, {\bf 62(1-2)}, 7--26.


\end{thebibliography}
\end{document}